\documentclass[journal,onecolumn]{IEEEtran}
\usepackage{longtable}
\usepackage{multirow}
\usepackage{tabularx}
\usepackage{lipsum}
\usepackage{multicol}
\usepackage{amssymb}
\usepackage{amsmath}
\usepackage{longtable}
\newtheorem{theorem}{Theorem}[section]
\newtheorem{lemma}[theorem]{Lemma}
\newtheorem{proposition}[theorem]{Proposition}
\newtheorem{corollary}[theorem]{Corollary}
\newtheorem{example}{Example}
\usepackage{xcolor}

\newenvironment{proof}[1][Proof]{\begin{trivlist}
\item[\hskip \labelsep {\bfseries #1}]}{\end{trivlist}}
\newenvironment{definition}[1][Definition]{\begin{trivlist}
\item[\hskip \labelsep {\bfseries #1}]}{\end{trivlist}}

\newenvironment{remark}[1][Remark]{\begin{trivlist}
\item[\hskip \labelsep {\bfseries #1}]}{\end{trivlist}}

\makeatletter

\newcommand{\Rmnum}[1]{\expandafter\@slowromancap\romannumeral #1@}
\makeatother

\ifCLASSINFOpdf
\else
\fi
\begin{document}
%
\title{Complementary Dual Algebraic Geometry Codes}

\author{ Sihem Mesnager$^1$ \and Chunming Tang$^2$ \and  Yanfeng Qi$^3$
\thanks{This work was supported by the SECODE European Project (Paris 8) and by
the National Natural Science Foundation of China
(Grant No. 11401480, 11531002). C. Tang
also acknowledges support from 14E013 and
CXTD2014-4 of China West Normal University.
Y. Qi also acknowledges support from Zhejiang provincial Natural Science Foundation of China (LQ17A010008).
}
\thanks{S. Mesnager is with Department of Mathematics, Universities of Paris VIII and XIII and Telecom ParisTech, LAGA, UMR 7539, CNRS, Sorbonne Paris Cit\'{e}. e-mail: smesnager@univ-paris8.fr}
\thanks{C. Tang is with School of Mathematics and Information, China West Normal University, Nanchong, Sichuan,  637002, China. e-mail: tangchunmingmath@163.com
}

\thanks{Y. Qi is with School of Science, Hangzhou Dianzi University, Hangzhou, Zhejiang, 310018, China.
e-mail: qiyanfeng07@163.com
}

}

%


\maketitle

\begin{abstract}

Linear complementary dual (LCD) codes  is a class of linear codes  introduced by Massey in 1964.
LCD codes have been extensively studied in literature recently. In addition to their applications in data storage, communications systems, and consumer electronics,  LCD codes have been employed in cryptography.
More specifically,  it has been shown
that LCD codes can also help improve the security of the information processed by sensitive devices, especially against so-called side-channel attacks (SCA) and fault non-invasive attacks. In this paper, we are  interested in the construction  of particular algebraic geometry (AG) LCD  codes which could be  good candidates to be resistant against SCA. We firstly provide a construction scheme  for obtaining LCD codes from elliptic curves. Then, some explicit LCD codes from elliptic curve are presented.
MDS codes are of the most importance in coding theory due to their theoretical significance and practical interests. In this paper,  all  the constructed LCD codes from elliptic curves are MDS or almost MDS. Some infinite  classes of LCD codes from elliptic curves are optimal due to the Griesmer bound. Finally, we introduce a  construction mechanism for obtaining LCD codes from any algebraic curve and derive some explicit LCD codes from hyperelliptic curves and Hermitian curves.
\end{abstract}

\begin{IEEEkeywords}
Linear complementary dual codes, algebraic geometry codes, algebraic curves,  elliptic curves, non-special divisors
\end{IEEEkeywords}

%
\IEEEpeerreviewmaketitle

\section{Introduction}
Linear complementary dual (LCD) cyclic codes over finite fields were first introduced and studied by Massey \cite{Mas} in
1964. In the literature LCD cyclic codes were referred to as reversible cyclic codes.
 It is well-known that LCD codes are asymptotically
good. Furthermore, using the full dimension spectra of linear
codes, Sendrier showed that LCD codes meet the asymptotic Gilbert-Varshamov bound \cite{Sen}.
Afterwards, LCD codes have been extensively studied in literature. In particular many  properties  and constructions of LCD codes have been obtained. Yang and Massey have provided in \cite{YM} a necessary and sufficient condition under which
a cyclic code  have a complementary dual.  Dougherty et al.  have
developed in \cite{DKO} a linear programming bound on the largest size of a LCD code of given length
and minimum distance. Esmaeili and Yari analyzed LCD codes that are quasi-cyclic
\cite{EY}.  Muttoo and Lal constructed a reversible code over $\mathbb F_q$ \cite{ML}. Tzeng and
Hartmann proved that the minimum distance of a class of reversible cyclic codes is greater
than the BCH bound \cite{TH}. In \cite{LDL16} Li et al.  studied a class of reversible BCH codes proposed in \cite{LDL} and extended the results on their parameters. As a byproduct, the parameters of some primitive BCH codes have been analyzed. Some of the obtained codes are optimal or have the best known parameters. In \cite{CG} Carlet and Guilley investigated an application of LCD codes
against side-channel attacks, and presented several constructions of LCD codes. In \cite{DLL}, Ding et al. constructed several families of LCD
cyclic codes over finite fields and analyzed their parameters.


Let  $K=\mathbb{F}_{q}$ be a finite field of order $q$ and $C/K$ be a  smooth projective curve of genus $g$.
  We denote by $D$ a
divisor  over $C/K$ : $D:=P_1+ \ldots+ P_n$, where $P_i (i=1, \cdots, n)$ are   pairwise  different places of degree one. Let
 $G$ be a divisor of  $C/K$ such that ${\mbox {supp }} D \cap {\mbox {supp }} G=\varnothing$. Let  $\cal {C} := {\cal
{C}}$$( D, G)$ be the associate  algebraic geometry (AG) code with the divisors  $D$ and $G$ defined as
\begin{equation*}
\mathcal {C}(D,G)=\{
(f(P_1),\ldots, f(P_n)),
f\in \mathcal{L}(G)\},
\end{equation*}
where ${\mathcal
{L}}(G)=\{f\in K(C), (f) \succeq -G\}\cup\{0\}$.
The code $\cal {C}$ is the image of ${\mathcal {L}}(G)$ under the evaluation map $ev_{D}$ given by
\begin{eqnarray*}
ev_{D} \, : \,\mathcal {L}(G) &\longrightarrow& {\mathbb {F}^n_{q}} \\
f & \longmapsto& (f(P_1), \ldots, f(P_n)).
\end{eqnarray*}

An algebraic geometry (AG) code $\mathcal{C}(D, G)$ associating with divisors $G$ and $D$ over the projective line is said to be \emph{rational}. In particular BCH codes and Goppa codes can be described by means of rational AG codes. All the generalized Reed-Solomon codes and extended generalized Reed-Solomon codes can be defined  under the framework of AG codes.

Recently,  it has been shown
that codes can also help improve the security of the information processed by sensitive devices, especially against  the so-called \emph{side-channel attacks} (SCA) and fault non-invasive attacks.

In this paper, we are  interested in the construction  of particular AG complementary dual (LCD) codes which can be resistant against SCA.
We firstly provide a construction scheme  for obtaining complementary dual codes from elliptic curves (Theorem \ref{gconstruction}). Then, some explicit complementary dual code are presented.
All  the constructed LCD codes from elliptic curve are MDS or almost MDS. Moreover, they contain some infinite  class of optimal codes meeting Griesmer bound on linear codes. Finally, we introduce a  construction mechanism for obtaining LCD codes from any algebraic
curve (Theorem \ref{gconstruction for curves} and Theorem \ref{Res:const}) and give some explicit LCD codes from hyperelliptic curve and Hermitian curves.
The constructed LCD codes presented in this paper could be good candidates of codes resistant against SCA, that is, codes having the property of being complementary dual codes with high minimal distance closed the Singleton bound.

This paper is organized as follows: In Section \ref{prel} we introduce the notations used in this paper and recall some basic facts about algebraic geometry codes. In section \ref{gene}, we present a general construction of
LCD codes from elliptic curves. In Section \ref{expl}, some explicit LCD codes from elliptic curves are derived. In Section \ref{comp}, we introduce a construction
mechanism for obtaining LCD codes from any algebraic curves  and give some explicit LCD codes from hyperelliptic curve and Hermitian curves.
\section{preliminaries}\label{prel}
In this section, we introduce notations and results  on
LCD codes,
algebraic geometry codes, and elliptic curves.
\subsection{Complementary dual codes and optimal codes}
A linear code of length $n$ over $\mathbb F_q$ is a linear subspace of $\mathbb F_q^n$. There is a canonical non-degenerate bilinear form on $\mathbb F_q^n \times \mathbb F_q^n$, defined by
$$ <(a_1, \cdots, a_n), (b_1, \cdots, b_n)>=a_1b_1+\cdots +a_nb_n.$$
For a linear code $\mathcal C$ of length $n$, the code
$$ \mathcal C^ \bot =\{\mathbf v \in \mathbb F_q^n: <\mathbf v, \mathbf c>=0 \text{ for any } \mathbf c \in \mathcal C\}$$
is called the dual of $\mathcal C$. The code  $\mathcal C^\bot $ is linear, and we have
$$dim_{\mathbb F_q}(\mathcal C)+ dim_{\mathbb F_q}(\mathcal C ^\bot)=n. $$
A linear code $\mathcal C$ is said to be a \emph{linear complementary dual (LCD) code} if the intersection with its dual $\mathcal C^\bot$ is trivial, that is, $\mathcal C \cap \mathcal C^\bot =\{0\}$.
The weight $wt(\mathbf v)$ of a vector $\mathbf v\in \mathbb F_q^n$ is the number of its nonzero coordinates. The minimum Hamming distance $d$ of a linear code $\mathcal C \neq \{0\}$ is defined by
$$d=\min \{wt(\mathbf c): \mathbf c \in \mathcal C\}.$$
An [n, k, d] linear code $\mathcal C$ is  a linear code of length $n$, dimension $k$ and minimum Hamming distance $d$.
We shall use the following codes.
\begin{definition}
Let $\mathbf a=(a_1,\cdots,a_n)$ with $a_i\in \mathbb F_{q}^\star $ and $\mathcal C \subseteq\mathbb{F}_q^n$. Then
\begin{align*}
\mathbf a\cdot \mathcal C:=\{(a_1c_1,\cdots,a_nc_n):(c_1,\cdots,c_n)\in \mathcal C\}.
\end{align*}
\end{definition}
Obviously, $\mathbf a \cdot \mathcal C$ is a linear code if and only if $\mathcal C$ is a linear code. These codes have the same dimension, minimum Hamming distance and weight distribution.

Let $n_q(k, d):= \min \{n: \text{ there is an } [n, k, d]~ \text{ linear code over } \mathbb F_q\}$. The $[n_q(k, d), k, d]$ codes
are called \emph{optimal codes}. The following result is known as the \emph{Griesmer bound} (see \cite{Dod} or \cite{HT}).
\begin{align}\label{Griesmer}
n_q(k,d)\ge g_q(k,d):=\sum_{i=0}^{k-1} \lceil \frac{d}{q^i} \rceil.
\end{align}

Under certain conditions on $d$ and $q$, it can be shown that $n_q(k, d)= g_q(k, d)$ (see \cite{Dod}
and \cite{HD}) for further references. Any $[g_q(k, d), k, d]$ code is optimal. Obviously,
$$n_q(k, d)\ge g_q(k, d)\ge  k + d-1.$$
The inequality $n_q(k, d)\ge  k + d -1$ is known as the \emph{Singleton bound}. If $d > q$, then
the Singleton bound is  always worse than the Greismer bound. The $[n, k, d]$ codes with
$n = k + d -1$ (resp. $n = k + d$) are called \emph{maximum distance separable (MDS) codes} (resp. \emph{almost maximum distance separable (MDS) codes}).
\subsection{Generalized algebraic geometry codes}
Let $C$ be a smooth projective curve of genus $g$. Throughout this paper, we assume $P_1,P_2,\cdots, P_n$ are pairwise different places   of $C$ of degree one and denote by $D$ the divisor $P_1+P_2+\cdots +P_n$.
We fix some notations which will be used throughout this paper.
\begin{itemize}
\item $\mathbb F_q$ denotes the finite field with $q=p^m$ elements;
\item $\mathrm{Tr}^m_1(x)=\sum_{i=0}^{m-1} x^{p^i}$ denotes the trace function from $\mathbb F_q$ to $\mathbb F_p$;
\item $C$ denotes a smooth projective curve over $\mathbb F_q$;
\item $\mathbb F_q(C)$ denotes the function field of $C$;
\item $C(\mathbb F_q)$ denotes the set of $\mathbb F_q$-rational points of  $C$;
\item $\Omega$ denotes the module of differentials of $C$;
\item $(f)$ denotes the principal divisor of $0\neq f\in \mathbb F_q(C)$;
\item $(\omega)$ denotes the divisor of differential $0\neq \omega\in \Omega$;
\item $v_P$ denotes the valuation of $\mathbb F_q(C)$ at the place $P$;
\item $\mathrm{Res}_P(\omega)$ denotes the residue of $\omega$ at $P$;
\item $G$ denotes a divisor of $C$ over $\mathbb F_q$;
\item $Supp(G)$ denotes the set of places in the support of $G$;
\item $\mathcal L(G):=\{f\in \mathbb F_q(C):(f)\succeq -G\} \cup \{0\}$;
\item $\Omega(G):=\{\omega\in \Omega:(\omega)\succeq G\}\cup \{0\}$;
\item $l(G)$ denotes the dimension of $\mathcal L(G)$ over $\mathbb F_q$;
\item $i(G)$ denotes the dimension of $\Omega(G)$ over $\mathbb F_q$.
\end{itemize}
Two divisor $D_1$ and $D_2$ are called equivalent, if there is a function $f\in \mathbb F_q(C)$ with $(f)=D_1-D_2$. Denote
two equivalent divisors $D_1$ and $D_2$ by
$D_1\sim D_2$.
The following famous result \cite{Sti}, known as the Riemann-Roch theorem is not only a central result  in algebraic geometry with applications in other areas, but it is also the key of several results in coding theory.
\begin{theorem}\label{Riemann-Roch}
Let $G$ be a divisor on a smooth projective curve of genus $g$ over $\mathbb F_q$. Then, for any Weil differential $0\neq \omega\in \Omega$
\begin{align*}
l(G) -i(G)=deg(G)+1-g \text{~~~~and~~~~} i(G)=l((\omega)-G).
\end{align*}
\end{theorem}

We call $i(G)$  the index of speciality of $G$. A divisor $G$ is called \emph{non-special} if $i(G)=0$ and otherwise it is called \emph{special}.
Note that $g-1$ is the least possible degree of a divisor of $G$ to be non-special, since  $0\le l(G)=deg(G)-g+1$. Moreover, if $deg(G)=g-1$, then $G$ is a non-special divisor if and only if $l(G)=0$. A non-special divisor of degree $g-1$ is never effective.

Let $G=\sum_{i=1}^n m_i P_i$ and $H=\sum_{i=1}^n m_i' P_i$ be two divisors. Then, we call $\sum_{i=1}^n \min(m_i,m_i')P_i$ the greatest common divisors denoted by $\mathrm{g.c.d}(G,H)$. Such a divisor is supported on the places common to the support of both divisors with coefficients the
minimum of those occurring in $G$ and $H$.
 We call $\sum_{i=1}^n \max(m_i,m_i')P_i$ the least multiple divisor denoted by $\mathrm{l.m.d}(G,H)$. Such a divisor is supported
on all the places in the supports of $G$ and $H$ with coefficients the maximum of
those occurring in the divisors $G$ and $H$.

For a divisor $G$ of $C$ with $v_{P_i}(G)=0(i=1,\cdots,n)$ and $2g-2<deg(G)<n$, and a vector $\mathbf a=(a_1,a_2,\cdots,a_n)$ with $a_i \in \mathbb F_q^\star $, we define a \emph{generalized algebraic geometry code}
\begin{align}
\mathcal {GC}(D,G,\mathbf a):=\{(a_1f(P_1),\cdots, a_nf(P_n)):f \in \mathcal L(G)\}.
\end{align}
If $\mathbf{a}=(1,1,\cdots,1)$, then,  $\mathcal{GC}(D,G,\mathbf a)$ is a classical algebraic geometry code denoted by $\mathcal{C}(D,G)$.
If $C$ is a curves of genus $1$ (called \emph{elliptic curves}), $\mathcal {GC}(D,G,\mathbf a)$ (resp. $\mathcal{C}(D,G)$) is called \emph{generalized elliptic code (resp. elliptic code)}.

Let $\omega $ be a \emph{Weil differential} such that $v_{P_i}(\omega)=-1$ for $i\in\{1,\cdots,n\}$. Then $\mathcal{C}(D, G)^{\bot}=\mathbf e \cdot \mathcal {C}(D, H)$ with
$H:=D-G+(\omega)$ and $\mathbf e=(\mathrm{Res}_{P_1}(\omega),
\cdots,\mathrm{Res}_{P_n}(\omega))$.  Thus,
\begin{align}\label{dual}
\mathcal{GC}(D, G,\mathbf a)^{\bot}=(\mathbf a^{-1}  \ast   \mathbf e) \cdot \mathcal {C}(D, H),
\end{align}
where $\mathbf a^{-1}  \ast   \mathbf e=(\frac{\mathrm{Res}_{P_i}(\omega)}{a_1},\cdots,\frac{\mathrm{Res}_{P_n}(\omega)}{a_n})$.\\
The following theorem determines the parameters of $\mathcal{GC}(D, G,\mathbf a)$ \cite{Sti}.
\begin{theorem}\label{[n,k,d]}
The code $\mathcal{GC}(D, G,\mathbf a)$ has dimension $k=deg(G)-g+1$ and minimum distance $d\ge n-deg(G)$.
\end{theorem}

From the definition of the generalized algebraic geometry codes we see that curves carrying many rational points may produce long codes.
 On the other hand, the number of $\mathbb F_q$-rational points of a smooth curve $C$ defined over $\mathbb F_q$ is bounded by the well known
 Hasse-Weil bound:
 $$|\#(C(\mathbb F_q))-(q+1)|\le 2 g \sqrt{q},$$
where $g$ is the geometric genus of $C$. As a consequence, curves attaining the bound (which are called \emph{maximal}) are particularly interesting in coding theory.
\subsection{Elliptic curves}
Let $E$ be an elliptic curve over $\mathbb{F}_q$ and $\mathcal O$ be the point at infinity of $E(\overline{\mathbb{F}}_q)$, where $\overline{\mathbb{F}}_q$
is the algebraic closure of $\mathbb{F}_q$ and $E(\overline{\mathbb{F}}_q)$ is the set of all points on $E$.  For any divisor $D\in Div(E)$, we denote $\overline{D}$ the unique rational point such that
$D- \overline D-(deg(D)-1) \mathcal O$ is a principal divisor. In fact, if $D=m_1P_1+m_2 P_2+\cdots + m_n P_n$, then $\overline D=m_1 P_1 \oplus m_2 P_2\oplus \cdots \oplus m_n P_n$, where $\oplus$ is the addition of points on the elliptic curve. For non-negative integer $r$, let $E[r]:=\{P\in E(\overline{\mathbb{F}}_q): \underbrace{P\oplus \cdots \oplus P}_{r}=\mathcal O\}$. We refer to \cite{Sil} for more details about elliptic curves.

\section{general constructions of LCD codes from elliptic curves}\label{gene}
In this section, we consider
the construction of LCD codes from
elliptic curves and determine the parameters of these LCD codes. We first present  a  proposition, which will be used in the following paper.
\begin{proposition}\label{a.C}
Let $a_i\in \mathbb F_q^\star$ $(i=1,\cdots,n)$ and $\mathcal C$ be a linear code in $\mathbb F_q^n$. If $\mathcal C ^ \bot= \mathbf e \cdot \mathcal C'$ with $\mathbf e=(a_1^2,\cdots,a_n^2)$ and $\mathcal C \cap \mathcal C'= \{0\}$, then, $(\mathbf a \cdot \mathcal C)^ \bot=\mathbf a \cdot  \mathcal C'$ and $\mathbf a \cdot \mathcal C$ is complementary dual, where $\mathbf a =(a_1,\cdots,a_n)$.
\end{proposition}
\begin{proof}
From $\mathcal C ^ \bot= \mathbf e \cdot \mathcal C'$, one has $dim_{\mathbb F_q}(\mathcal C)+  dim_{\mathbb F_q}(\mathcal C')=n$, and, for any $(c_1,\cdots,c_n)\in \mathcal C$ and $(c_1',\cdots,c_n')\in \mathcal C'$,
$a_1c_1\cdot a_1c_1'+\cdots +a_nc_n\cdot a_nc_n'=0$. Thus, $(\mathbf a \cdot \mathcal C)^ \bot=\mathbf a \cdot  \mathcal C'$.

Suppose $(a_1\overline{c}_1,\cdots, a_n\overline{c}_n) \in (\mathbf a \cdot \mathcal C)^ \bot \cap \mathbf a \cdot \mathcal C$, where $(\overline{c}_1,\cdots, \overline{c}_n) \in\mathcal C$. Then, for any $(c_1,\cdots,c_n) \in \mathcal C$, $a_1c_1\cdot a_1\overline{c}_1+\cdots +a_nc_n\cdot a_n\overline{c}_n=0$. Thus, $(a_1^2\overline{c}_1,\cdots, a_n^2 \overline{c}_n)\in \mathcal C^ \bot$ and $(\overline{c}_1,\cdots,\overline{c}_n)\in \mathcal C'$. From  $\mathcal C \cap \mathcal C'= \{0\}$, we obtain $(\overline{c}_1,\cdots,\overline{c}_n)=(0,\cdots,0)$. Hence, $\mathbf a \cdot \mathcal C$ is complementary dual. This completes the proof.
\end{proof}

The following theorem constructs the LCD codes from elliptic curves and determines
the corresponding parameters of these codes.
\begin{theorem}\label{gconstruction}
Let $E$ be an elliptic curve over $\mathbb{F}_q$ and  $G$, $D=P_1+P_2+\cdots +P_n$ be two divisors over $E$, where   $0<deg(G)<n$.  Let $\omega$ be a Weil differential such that $(w)=G+H-D$ for some divisor $H$ and $Supp(G)\cap Supp(D)=Supp(H)\cap Supp(D)=\emptyset$.     Assume that

{\rm (1)} There is a vector $\mathbf a=(a_1,\cdots,a_n)\in (\mathbb{F}_q^\star)^n$ with $\mathrm{Res}_{P_i}(\omega)=a_i^2$;

{\rm (2)} $deg(\mathrm{g.c.d}(G,H))=0$;

{\rm (3)} $\mathrm{g.c.d}(G,H)$  is not a principal divisor.

Then, $\mathcal{ GC}(D,G,\mathbf a)$ is a LCD code with dimension $deg(G)$ and minimum distance $d\ge n-deg(G)$,  and  the dual code $\mathcal{ GC}(D,H,\mathbf a)$ of $\mathcal{ GC}(D,G,\mathbf a)$  is a LCD code with dimension $n-deg(G)$ and minimum distance $d^\bot \ge deg(G)$.
\end{theorem}
\begin{proof}
Note that a canonical divisor over an  elliptic curve is a principal divisor. Since  $\mathrm{g.c.d}(G,H)$  is not a principal divisor and $\mathrm{l.m.d}(G,H)-D=(G+H-D)-
\mathrm{g.c.d}(G,H)$, then $\mathrm{l.m.d}(G,H)-D$ is not a principal divisor.

We firstly prove that $\mathcal{C}(D,G) \cap \mathcal{C}(D,H)=\{0\}$. Suppose that there are some $f\in \mathcal{L}(G)$ and some
$g\in \mathcal{L}(H)$ such that $f(P_i)=g(P_i)$ for $i=\{1,2,\cdots ,n\}$. Consider the following two mutually exclusive cases on
$h:=f-g$

{\rm 1)} Case $h=0$. One has $f\in \mathcal{L}(G) \cap \mathcal{L}(H)$.  Then, $f\in \mathcal{L}(\mathrm{g.c.d(G,H)})$. Since $\mathrm{g.c.d(G,H)}$
is not a principal divisor, then $f\in \mathcal{L}(\mathrm{g.c.d(G,H)})=\{0\}$ and  $f=g=0$.

{\rm 2)} Case $h\neq 0$. One has $h \in \mathcal{L}(\mathrm{l.m.d}(G,H)-D)$ as $h(P_i)=0 (i=1,\cdots,n)$. Since $\mathrm{l.m.d(G,H)}-D$
is not a principal divisor, then $h\in \mathcal{L}(\mathrm{l.m.d(G,H)}-D)=\{0\}$  and  $h=0$, which is a contradiction.

Hence, $\mathcal{C}(D,G) \cap \mathcal{C}(D,H)=\{0\}$. From Equation (\ref{dual}), $\mathcal{C}(D,G)^\bot =\mathbf e \cdot \mathcal{C}(D,H)$ with  $\mathbf{e}=(\mathrm{Res}_{P_1}(\omega),\cdots, \mathrm{Res}_{P_n}(\omega))$,  and Proposition \ref{a.C}, $\mathcal{ GC}(D,G,\mathbf a)$ and $\mathcal{ GC}(D,H,\mathbf a)$  are complementary dual codes with $\mathcal{ GC}(D,G,\mathbf a)^\bot =\mathcal{ GC}(D,H,\mathbf a)$. The dimensions and minimum distances follow from Theorem \ref{[n,k,d]}.
\end{proof}

\begin{remark}\label{optimal}
An interesting result of Cheng \cite{Che} says that the minimum distance problem is already $\mathbf{NP}$-hard (under $\mathbf{RP}$-reduction) for general elliptic curves codes.
In \cite{LWZ}, Li et al.  showed that the minimum distance of algebraic codes from elliptic curves also has a simple explicit formula if the evaluation set is suitably large (at least $\frac{2}{3}$ of the group order). This method proves that, if $n=\#D \geq q+2$ and $3<deg(G)<q-1$, then, $\mathcal{GC}(D,G,\mathbf a)$ has the deterministic minimum distance $n-deg(G)$. In  this cases, $\mathcal{GC}(D,G,\mathbf a)$ has parameters $[n, deg(G), n-deg(G)]$, where $n= \# D$.
Thus, $\mathcal{GC}(D,G,\mathbf a)$ is an almost MDS code. Let $n\ge q+3$. From \cite{Jan},  if $2 \le deg(G)\le n-(q+1)$ or $q+1 \le deg(G)\le n-2$,
the elliptic code $\mathcal{GC}(D,G,\mathbf a)$ is  optimal.  Hence, many different (perhaps  nonequivalent) LCD generalized elliptic optimal codes exist, since many elliptic curves with more than $q+2$ rational points exist.
\end{remark}

From Theorem \ref{gconstruction}, we have the following two corollaries.
\begin{corollary}\label{(r-1)O+r Q}
Let $n$, $r$ be positive integers with $2\le r \le \frac{n+1}{2}$ and $D=P_1+\cdots +P_n$ be a divisor  such that $ \overline{D}\not \in E[r-1]$
and $\mathcal O,\overline D\not \in Supp(D)$. Let $G=(r-1)\mathcal O+r \overline D$ and $H=(n-r)\mathcal O-(r-1)\overline D$. Let $\omega$ be the Weil  differential such that
$(\omega)=(n-1) \mathcal{O}+\overline D-D$. Assume that there is a vector $\mathbf a=(a_1,\cdots,a_n)\in (\mathbb{F}_q^\star)^n$ with $\mathrm{Res}_{P_i}(\omega)=a_i^2$. Then, $\mathcal{ GC}(D,G,\mathbf a)$ is a LCD code with dimension $2r-1$ and minimum distance $d\ge n-2r+1$,  and  the dual code $\mathcal{ GC}(D,H,\mathbf a)$ of $\mathcal{ GC}(D,G,\mathbf a)$  is a LCD code with dimension $n-2r+1$ and minimum distance $d^\bot \ge 2r-1$.
\end{corollary}
\begin{proof}
Note that $G+H-D=(n-1) \mathcal{O}+\overline D-D$ is a principal divisor. Then, there exists a Weil differential $\omega$ such that $(\omega)=(n-1) \mathcal{O}+\overline D-D$. From $\overline D\not \in E[r-1]$, $\mathrm{g.c.d}(G,H)=(r-1)\mathcal O-(r-1) \overline D$ is not a principal divisor.
This corollary follows from  Theorem \ref{gconstruction}.
\end{proof}

\begin{corollary}\label{rO+rQ}
Let $Q$ be a place on $E$ different from $\mathcal{O}$, $G=(r\cdot deg(Q)) \mathcal{O}+ r Q$, and $D=P_1+\cdots +P_n$ be a principal divisor, where $0<2r\cdot deg(Q)< n$, $\underbrace{\overline{Q}\oplus \cdots \oplus \overline{Q}} _{r}\neq \mathcal{O}$, and  $Q,\mathcal{O} \not \in Supp(D)$.  Let $\omega$ be the Weil differential such that
$(\omega)=n \mathcal{O}-D$. Assume that there is a vector $\mathbf a=(a_1,\cdots,a_n)\in (\mathbb{F}_q^\star)^n$ with $\mathrm{Res}_{P_i}(\omega)=a_i^2$. Then, $\mathcal{ GC}(D,G,\mathbf a)$ is a LCD code with dimension $2r\cdot deg(Q)$ and minimum distance $d\ge n-2r\cdot deg(Q)$,  and  the dual code $\mathcal{ GC}(D,H,\mathbf a)$ of $\mathcal{ GC}(D,G,\mathbf a)$  is a LCD code with dimension $n-2r\cdot deg(Q)$ and minimum distance $d^\bot \ge 2r\cdot deg(Q)$, where $H=(n-r\cdot deg(Q))\mathcal{O}-r Q$.
\end{corollary}
\begin{proof}
Note that $(\omega)=G+H-D=n \mathcal O-D$ and $\mathrm{g.c.d}(G,H)=(r\cdot deg(Q)) \mathcal O-rQ$. From $\underbrace{\overline{Q}\oplus \cdots \oplus \overline{Q}} _{r}\neq \mathcal{O}$, $\mathrm{g.c.d}(G,H)$ is not a principal divisor. The corollary follows   from Theorem \ref{gconstruction}.
\end{proof}

\section{explicit construction of LCD codes from elliptic curves}\label{expl}
The previous results presented in  Section \ref{gene} are of significance if there are interesting examples of elliptic curves and divisors $G,H,D$ satisfying the properties assumed in Theorem \ref{gconstruction}. Hence, in this section we present general examples,  where the assumptions of Theorem \ref{gconstruction} are satisfied.

Let $q=2^m$ and $E^{(0)}$ be an elliptic curve defined by the equation
\begin{align}\label{E^0}
y^2+y=x^3+bx+c,
\end{align}
where $b,c\in \mathbb F_q$. The point at infinity is denoted by $\mathcal O$.
Let $S$ be the set of $x$-components of  the affine points of $E^{(0)}$ over $\mathbb F_q$, that is,
\begin{align*}
S:=\{\alpha\in \mathbb F_q: \text{ there is } \beta\in \mathbb F_q \text{ such that } \beta^2+\beta=\alpha^3+b\alpha+c\}.
\end{align*}
For any $\alpha\in S$, there exactly exist two points with $x$-component $\alpha$. Denote these two points
corresponding to $\alpha$ by $P_{\alpha}^{+}$ and $P_{\alpha}^{-}$. Then the set $E^{(0)}(\mathbb F_q)$ of  all rational points of $E^{(0)}$ over $\mathbb F_q$ is  $E^{(0)}(\mathbb F_q)=\{P_{\alpha}^{+}: \alpha \in S\}\cup \{P_{\alpha}^{-}: \alpha \in S\} \cup \{\mathcal O\}$.
The following Lemma can be found in \cite{Has}.
\begin{lemma}\label{E^0}
Let $s$ be a positive integer, $\{\alpha_1,\cdots, \alpha_s\}$ be a subset of $ S$ with cardinality $s$,  and  $D=\sum _{i=1}^s (P_{\alpha_i}^+ + P_{\alpha_i}^-)$. Let $h=\prod_{i=1}^{s} (x+\alpha_i)$ and $\omega=\frac{d x}{h}$. Then, $(\omega)=2s\cdot \mathcal O-D$
and
\[\mathrm{Res}_{P_{\alpha_j}^{+}}(\omega)=\mathrm{Res}_{P_{\alpha_j}^{-}}(\omega)   =\frac{1}{\prod _{i=1,i\neq j}^{s} (\alpha_j+\alpha_i)},\]
for any $j\in \{1,2,\cdots,s\}$.
\end{lemma}
The following result is a direct consequence of Corollary \ref{rO+rQ} and Lemma \ref{E^0}.
\begin{theorem}\label{E^0:rO+rQ}
Let $s$ be a positive integer, $\{\alpha_0,\alpha_1,\cdots, \alpha_s\}$ be a subset of $ S$ with cardinality $s+1$, $D= (P_{\alpha_1}^+ + P_{\alpha_1}^-)+\cdots+(P_{\alpha_s}^+ + P_{\alpha_s}^-)$
and $G=r\mathcal O+r P_{\alpha_0}^+$, where $0<r<s$ and $P_{\alpha_0}^+\not \in E^{(0)}[r]$. Let $b_j=\frac{1}{\prod _{i=1,i\neq j}^{s} (\alpha_j^{2^{m-1}}+\alpha_i^{2^{m-1}})}$ $(j=1,\cdots,s)$
and $\mathbf a=(b_1,b_1,\cdots,b_s,b_s)$. Then, $\mathcal{ GC}(D,G,\mathbf a)$ is a LCD code with dimension $2r$ and minimum distance $d\ge 2(s-r)$,  and  the dual code $\mathcal{ GC}(D,H,\mathbf a)$ of $\mathcal{ GC}(D,G,\mathbf a)$  is a LCD code with dimension $2(s-r)$ and minimum distance $d^\bot \ge 2r$, where $H=(2s-r)\mathcal{O}-r P_{\alpha_0}^+$.
\end{theorem}
\begin{example}
Let $q=2^4$, $\mathbb{F}_q=\mathbb{F}_2[\rho]$ with $\rho^4+\rho+1=0$, and $E^{(0)}$ be the elliptic curve defined by $y^2+y=x^3+\rho^3$. Let $P^+=(\rho,0)$
, $P^-=(\rho,1)$ and $D=E^{(0)}(\mathbb F_q) \backslash \{\mathcal O, P^+, P^-\}$. Then $4 P^+ \neq \mathcal O$ and $\# D=22$.
Let  $G=4\mathcal O+4P^+$ and $H=18\mathcal O-4P^+$. Then, $\mathcal{ GC}(D,G,\mathbf a)$ in Theorem \ref{E^0:rO+rQ} is a LCD code with parameters [22, 8, 14], and  the dual code $\mathcal{ GC}(D,H,\mathbf a)$ in Theorem \ref{E^0:rO+rQ}  of $\mathcal{ GC}(D,G,\mathbf a)$  is a LCD code with parameters [22, 14, 8], which is verified by  MAGMA.
\end{example}

\begin{corollary}\label{E^0:rO+rQ(coro)}
Let $N=\# E^{(0)}(\mathbb F_q)$ and $s$ be a positive integer, $\{\alpha_0,\alpha_1,\cdots, \alpha_s\}$ be a subset of $ S$ with cardinality $s+1$, $D= (P_{\alpha_1}^+ + P_{\alpha_1}^-)+\cdots+(P_{\alpha_s}^+ + P_{\alpha_s}^-)$ and $G=r\mathcal O+r P_{\alpha_0}^+$, where $0<r<s$ and $g.c.d(r,N)=1$. Let $b_j=\frac{1}{\prod _{i=1,i\neq j}^{s} (\alpha_j^{2^{m-1}}+\alpha_i^{2^{m-1}})}$ $(j=1,\cdots,s)$
and $\mathbf a=(b_1,b_1,\cdots,b_s,b_s)$. Then, $\mathcal{ GC}(D,G,\mathbf a)$ is a LCD code with dimension $2r$ and minimum distance $d\ge 2(s-r)$,  and  the dual code $\mathcal{ GC}(D,H,\mathbf a)$ of $\mathcal{ GC}(D,G,\mathbf a)$  is a LCD code with dimension $2(s-r)$ and minimum distance $d^\bot \ge 2r$, where $H=(2s-r)\mathcal{O}-r P_{\alpha_0}^+$.
\end{corollary}
\begin{proof}
From $g.c.d(r,N)=1$, there exist integers $k_1$ and $k_2$ such that $k_1r+k_2N=1$. Then, $(k_1r) P_{\alpha_0}^+=P_{\alpha_0}^+ \ominus (k_2N) P_{\alpha_0}^+=P_{\alpha_0}^+\neq \mathcal O$. Thus, $P_{\alpha_0}^+\not \in E^{(0)}[r]$.
This corollary follows from Theorem \ref{E^0:rO+rQ}.
\end{proof}
\begin{remark}
The following Table I lists the numbers of rational points of
some elliptic curves over
$\mathbb{F}_{2^m}$ \cite{Men}. For general elliptic curves over a finite field, we can use the Schoof's algorithms \cite{Sch} to count the number of rational  points.

\begin{table}[htbp]\label{tab1}
\newcommand{\tabincell}[2]{\begin{tabular}{@{}#1@{}}#2\end{tabular}}
\centering
 \caption{the numbers of rational points of elliptic curves
 over $\mathbb{F}_{q} (q=2^m)$}
\begin{tabular}{|c|c|c|}
\hline
Elliptic Curve $E^{(0)}$ & $m$ & $\#E(\mathbb{F}_{2^m})$\\
\hline
$y^2+y
 =x^3$ & \tabincell{c}{ odd $m$ \\
 $m\equiv 0 \mod 4$ \\
 $m\equiv 2\mod 4$} & \tabincell{c}{
 $q+1$ \\ $q+1-2\sqrt{q}$ \\$
 q+1+2\sqrt{q}$}\\\hline
$y^2+y=x^3+x$ & \tabincell{c}{
$m\equiv 1,7 \mod 8$\\
$m\equiv 3,5\mod 8$} & \tabincell{c}{
$q+1+\sqrt{2q}$\\ $q+1-\sqrt{2q}$} \\\hline
$y^2+y=x^3+x+1$ & \tabincell{c}{
$m\equiv 1,7 \mod 8$\\
$m\equiv 3,5\mod 8$} & \tabincell{c}{
$q+1-\sqrt{2q}$\\ $q+1+\sqrt{2q}$} \\\hline
$y^2+y=x^3+\delta x
~(\mathrm{Tr}^m_1(\delta)=1)$ &
even $m$ & $q+1$\\\hline
$y^2+y=x^3+\omega~
(\mathrm{Tr}^m_1(\omega)=1)$ & \tabincell{c}{
$m\equiv 0\mod 4$\\
$m\equiv 2\mod 4$} & \tabincell{c}{
$q+1+2\sqrt{q}$\\ $q+1-2\sqrt{q}$} \\\hline
\end{tabular}
\end{table}
\end{remark}

\begin{theorem}\label{E^0:(r+1)O+rQ}
Let $s$ be a positive integer, $\{\alpha_1,\cdots, \alpha_s\}$ be a subset of $ S$ with cardinality $s$, $D=P_{\alpha_1}^- + (P_{\alpha_2}^+ + P_{\alpha_2}^-)+\cdots+(P_{\alpha_s}^+ + P_{\alpha_s}^-)$ and $G=(r+1)\mathcal O+r P_{\alpha_1}^+$, such that $0\le r<s-1$ and $P_{\alpha_1}^+\not \in E^{(0)}[r+1]$. Let $b_j=\frac{1}{\prod _{i=1,i\neq j}^{s} (\alpha_j^{2^{m-1}}+\alpha_i^{2^{m-1}})}$ $(j=1,\cdots,s)$
and $\mathbf a=(b_1,b_2,b_2\cdots,b_s,b_s)$. Then, $\mathcal{ GC}(D,G,\mathbf a)$ is a LCD code with dimension $2r+1$ and minimum distance $d\ge 2(s-r-1)$,  and  the dual code $\mathcal{ GC}(D,H,\mathbf a)$ of $\mathcal{ GC}(D,G,\mathbf a)$  is a LCD code with dimension $2(s-r-1)$ and minimum distance $d^\bot \ge 2r+1$, where $H=(2s-r-1)\mathcal{O}-(r+1) P_{\alpha_1}^+$.
\end{theorem}
\begin{proof}
From Lemma \ref{E^0}, $(\omega)=G+H-D$, where $\omega=\frac{1}{\prod _{i=1}^{s} (x+\alpha_i)}dx$. Note that $\mathrm{g.c.d}(G,H)=(r+1)\mathcal O-(r+1)P_{\alpha_1}^+$. This theorem follows from  $P_{\alpha_1}^+\not \in E^{(0)}[r+1]$ and Theorem \ref{gconstruction}.
\end{proof}
\begin{example}
Let $q=2^4$, $\mathbb{F}_q=\mathbb{F}_2[\rho]$ with $\rho^4+\rho+1=0$ and $E^{(0)}$ be the elliptic curve defined by $y^2+y=x^3+\rho^3$. Let $P^+=(\rho,0)$ and $D=E^0(\mathbb F_q) \backslash \{\mathcal O, P^+\}$. Then $4 P^+ \neq \mathcal O$ and $\# D=23$.
Let  $G=4\mathcal O+3P^+$ and $H=20\mathcal O-4P^+$. Then, $\mathcal{ GC}(D,G,\mathbf a)$ in Theorem \ref{E^0:(r+1)O+rQ} is a LCD code with parameters [23, 7, 16], and  the dual code $\mathcal{ GC}(D,H,\mathbf a)$ in Theorem \ref{E^0:(r+1)O+rQ}  of $\mathcal{ GC}(D,G,\mathbf a)$  is a LCD code with parameters [23, 16, 7], which is verified by MAGMA.
\end{example}

\begin{corollary}\label{E^0:(r+1)O+rQ(coro)}
Let $N=\# E^{(0)}(\mathbb F_q)$ and $r, s$ be positive integers, where $0\le r<s-1$ and $g.c.d(r+1,N)=1$. Let $\{\alpha_1,\cdots, \alpha_s\}$ be a subset of $ S$ with cardinality $s$, $D=P_{\alpha_1}^- + (P_{\alpha_2}^+ + P_{\alpha_2}^-)+\cdots+(P_{\alpha_s}^+ + P_{\alpha_s}^-)$ and $G=(r+1)\mathcal O+r P_{\alpha_1}^+$. Let $b_j=\frac{1}{\prod _{i=1,i\neq j}^{s} (\alpha_j^{2^{m-1}}+\alpha_i^{2^{m-1}})}$ $(j=1,\cdots,s)$
and $\mathbf a=(b_1,b_2,b_2\cdots,b_s,b_s)$. Then, $\mathcal{ GC}(D,G,\mathbf a)$ is a LCD code with dimension $2r+1$ and minimum distance $d\ge 2(s-r-1)$,  and  the dual code $\mathcal{ GC}(D,H,\mathbf a)$ of $\mathcal{ GC}(D,G,\mathbf a)$  is a LCD code with dimension $2(s-r-1)$ and minimum distance $d^\bot \ge 2r+1$, where $H=(2s-r-1)\mathcal{O}-(r+1) P_{\alpha_1}^+$.
\end{corollary}
\begin{proof}
The result follows from Theorem \ref{E^0:(r+1)O+rQ} and similar arguments  used in the proof of Corollary \ref{E^0:rO+rQ(coro)}.
\end{proof}

\begin{theorem}\label{E: non-special}
Let $r,s$ be   integers with $0\le r<\frac{s-2}{2}$,  $\{\alpha_0, \alpha_1,\cdots, \alpha_s\}$ be a subset of $ S$ with cardinality $s+1$,  $D=\sum _{i=1}^s (P_{\alpha_i}^+ + P_{\alpha_i}^-)$ and $G=(2r+3)\cdot \mathcal O+ r \cdot (P_{\alpha_0}^+ +P_{\alpha_0}^-)+   P_{\alpha_0}^+ $. Let $b_j=\frac{1}{\prod _{i=0,i\neq j}^{s} (\alpha_j^{2^{m-1}}+\alpha_i^{2^{m-1}})}$ $(j=1,\cdots,s)$
and $\mathbf a=(b_1,b_1,b_2,b_2\cdots,b_s,b_s)$. Then, $\mathcal{ GC}(D,G, \mathbf a)$ is a LCD code with dimension $4(r+1)$ and minimum distance $d\ge 2s-4(r+1)$,  and  the dual code $\mathcal{ GC}(D,H,\mathbf a)$ of $\mathcal{ GC}(D,G, \mathbf a)$  is a LCD code with dimension $2s-4(r+1)$ and minimum distance $d^\bot \ge 4(r+1)$, where $H=( 2s-2r-1)\cdot \mathcal O-(r +2) \cdot (P_{\alpha_0}^+ +P_{\alpha_0}^-)+ P_{\alpha_0}^- $.
\end{theorem}
\begin{proof}
Let $h=\prod_{i=0}^{s} (x+\alpha_i)$ and $\omega=\frac{1}{h} dx$. Then, from Lemma \ref{E^0}, $(\omega)=G+H-D. $  Note that $\mathrm{g.c.d}(G,H)=(2r+3)\mathcal O-(r+2)(P_{\alpha_0}^+ +
P_{\alpha_0}^-)+ P_{\alpha_0}^-$. From $P_{\alpha_0}^+\oplus P_{\alpha_0}^-=\mathcal O$, $\mathrm{g.c.d}(G,H)\sim P_{\alpha_0}^--\mathcal O$. Thus, $\mathrm{g.c.d}(G,H)$ is not a principal divisor. This theorem follows from Theorem \ref{gconstruction}.
\end{proof}

\begin{example}
Let $q=2^4$, $\mathbb{F}_q=\mathbb{F}_2[\rho]$ with $\rho^4+\rho+1=0$ and $E^{(0)}$ be the elliptic curve defined by $y^2+y=x^3+\rho^3$. Let $P^+=(\rho,0)$
, $P^-=(\rho,1)$ and $D=E^0(\mathbb F_q) \backslash \{\mathcal O, P^+, P^-\}$. Then $\# D=22$.
Let  $G=3\mathcal O+P^+$ and $H=21\mathcal O-2(P^+ +P^-)+P^-$. Then, $\mathcal{ GC}(D,G,\mathbf a)$ in Theorem \ref{E: non-special} is a LCD code with parameters [22, 4, 18], and  the dual code $\mathcal{ GC}(D,H,\mathbf a)$ in Theorem \ref{E: non-special}  of $\mathcal{ GC}(D,G,\mathbf a)$  is a LCD code with parameters [22, 18, 4], which is verified by  MAGMA. From the Remark below\ref{gconstruction}, both $\mathcal{ GC}(D,G,\mathbf a)$ and $\mathcal{ GC}(D,H,\mathbf a)$ are optimal.
\end{example}

\section{LCD codes from general algebraic curves}\label{comp}
In this section we consider the construction of LCD  codes
from any algebraic curves, present
a construction mechanism of LCD codes from algebraic geometry codes, and give concrete construction of
LCD codes from hyperelliptic curves and
Hermitian curves.

Two theorems on constructing
LCD codes from algebraic curves are given below.
\begin{theorem}\label{gconstruction for curves}
Let $C$ be a smooth  projective curve of genus $g$ over $\mathbb{F}_q$ and  $G$, $D=P_1+P_2+\cdots +P_n$ be
two divisors over $C$, where $2g-2<deg(G)<n$.  Let $\omega$ be a Weil differential such that $(w)=G+H-D$ for some divisor $H$ and $Supp(G)\cap Supp(D)=Supp(H)\cap Supp(D)=\emptyset$.    Assume that

{\rm (1)} There is a vector $\mathbf a=(a_1,\cdots,a_n)\in (\mathbb{F}_q^\star)^n$ with $\mathrm{Res}_{P_i}(\omega)=a_i^2$;

{\rm (2)} $\mathrm{g.c.d}(G,H)$ is a non-special divisor of degree $g-1$.

Then, $\mathcal{ GC}(D,G,\mathbf a)$ is a LCD code with dimension $deg(G)+1-g$ and minimum distance $d\ge n-deg(G)$,  and  the dual code $\mathcal{ GC}(D,H,\mathbf a)$ of $\mathcal{ GC}(D,G,\mathbf a)$  is a LCD code with dimension $n-deg(G)-1+g$ and minimum distance $d^\bot \ge deg(G)+2-2g$.
\end{theorem}
\begin{proof}
We first prove that $\mathcal{C}(D,G) \cap \mathcal{C}(D,H)=\{0\}$. Suppose that there exist some $f\in \mathcal{L}(G)$ and
$g\in \mathcal{L}(H)$ such that $f(P_i)=g(P_i)$ for $i=\{1,2,\cdots ,n\}$. Consider the following two mutually exclusive cases on
$h:=f-g$

{\rm 1)} Case $h=0$. One has $f\in \mathcal{L}(G) \cap \mathcal{L}(H)$ and $f\in \mathcal{L}(\mathrm{g.c.d(G,H)})$. Since  $\mathrm{g.c.d}(G,H)$ is a non-special divisor,  $\mathcal{L}(\mathrm{g.c.d(G,H)})=\{0\}$ and $f=g=0$.

{\rm 2)} Case $h\neq 0$. One has $h \in \mathcal{L}(\mathrm{l.m.d}(G,H)-D)$ as $h(P_i)=0 (i=1,\cdots,n)$.
Let $k=l(\mathrm{l.m.d}(G,H)-D)-l(\mathrm{g.c.d}(G,H))$. From  $\mathrm{g.c.d}(G,H)= (\omega)-(\mathrm{l.m.d}(G,H)-D)$ and Theorem \ref{Riemann-Roch},
\begin{align*}
k=&l(\mathrm{l.m.d}(G,H)-D)-l((\omega)-(\mathrm{l.m.d}(G,H)-D))\\
=& deg(\mathrm{l.m.d}(G,H)-D)+1-g\\
=& deg((\omega)- \mathrm{g.c.d}(G,H))+1-g\\
=&0.
\end{align*}
From $\mathcal{L}(\mathrm{g.c.d}(G,H))=\{0\}$, $ \mathcal{L}(\mathrm{l.m.d}(G,H)-D)=\{0\}$.
Thus, $h=0$, which is a contradiction.

Hence, $\mathcal{C}(D,G) \cap \mathcal{C}(D,H)=\{0\}$. From Equation (\ref{dual}), $\mathcal{C}(D,G)^\bot =\mathbf e \cdot \mathcal{C}(D,H)$ with $\mathbf{e}=(\mathrm{Res}_{P_1}(\omega),\cdots, \mathrm{Res}_{P_n}(\omega))$ and Proposition \ref{a.C}, $\mathcal{ GC}(D,G,\mathbf a)$ and $\mathcal{ GC}(D,H,\mathbf a)$  are complementary dual codes with $\mathcal{ GC}(D,G,\mathbf a)^\bot =\mathcal{ GC}(D,H,\mathbf a)$. The dimensions and minimum distances follow from Theorem \ref{[n,k,d]}.
\end{proof}
\begin{remark}
In \cite{BL}, S. Ballet and D. Le Brigand proved that if $\# C(\mathbb F_q) \ge g+1$, there exists a non-special divisor such that $deg(G)=g-1$ and $Supp(G)\subset C(\mathbb F_q)$. Then, the existence of non-special divisors of degree $g-1$ is often clear since the involved algebraic curves  have many rational points. However, the problem lies in their effective determination. Moreover, actually almost all the divisors with degree $g-1$ are non-special (the terminology almost all means all but finitely many).

\end{remark}

\begin{example}
Let $q=2$ and $C$ be the   projective curve
of genus 1 defined by $Y^2Z+YZ^2=X^{3}$ over $\mathbb F_{4}=\{0, 1, \rho, \rho^2\}$.
Then, $C(\mathbb F_4)=\{\mathcal O, Q, P_1,\cdots, P_7\}=\{(0 : 1 : 0), (0 : 0 : 1), (0 : 1 : 1), (\rho : \rho : 1), (\rho : \rho^2 : 1), (\rho^2 : \rho : 1), (\rho^2 : \rho^2 : 1), (1 : \rho : 1), (1 : \rho^2 : 1) \}$. Let $D=\{P_1,\cdots, P_7\}$, $G=2\mathcal O+Q$ and $H=6\mathcal O-2Q$. Then, $\mathrm{g.c.d}(G,H)=2\mathcal O -2 Q$ is non-special, $(\frac{Z^4}{X^4+Z^3X} d\frac{X}{Z})=G+H-D$ and $\mathrm{Res}_{P_i}(\frac{Z^4}{X^4+Z^3X} d\frac{X}{Z})=1$ for $i\in \{1, \cdots, 7\}$. Note that  $(\frac{X}{Z})=Q-2\mathcal O+P_1$ and $(\frac{Y+Z}{X})=-Q-\mathcal O+2P_1$.
Thus, $\{1, \frac{X}{Z}, \frac{Y+Z}{X}\}$ is a basis of $\mathcal L(2\mathcal O +Q)$. Evaluate the functions in $\{1, \frac{X}{Z}, \frac{Y+Z}{X}\}$  at the places $\{P_1, \cdots, P_7\}$. One obtains the generator matrix of $\mathcal{C}(D,G)$
\begin{equation*}
\left(
\begin{array}{ccccccc}
   1   &  1  &   1   &  1  &   1  &   1   &  1\\
   0  & \rho &  \rho & \rho^2 & \rho^2  &   1  &   1\\
    0  & \rho  &   1   &  1 & \rho^2 &  \rho^2 &  \rho
\end{array}
\right)
\end{equation*}
Moreover, $\mathcal{C}(D,G)$ is a LCD code with parameters [7, 3, 4],  which is verified by MAGMA. This code is optimal.
\end{example}

\begin{theorem}\label{Res:const}
Let $C$ be a smooth  projective curve of genus $g$ over $\mathbb{F}_q$ and  $G$, $D=P_1+P_2+\cdots +P_n$ be two divisors over $C$, where $2g-2<deg(G)<n$.  Let $\omega$ be a Weil differential such that $(w)=G+H-D$ for some divisor $H$ and $Supp(G)\cap Supp(D)=Supp(H)\cap Supp(D)=\emptyset$.    Assume that

{\rm (1)} $\mathrm{Res}_{P_i}(\omega)=\mathrm{Res}_{P_j}(\omega)$ for $1\le i< j \le n$;

{\rm (2)} $\mathrm{g.c.d}(G,H)$ is a  divisor of degree $g-1$.

Then, $\mathcal{ C}(D,G)$ is a LCD code if and only if $\mathrm{g.c.d}(G,H)$ is non-special.
\end{theorem}
\begin{proof}
Let $\mathrm{Res}_{P_i}=c$ for $1\le i\le n$. From Equation (\ref{dual}), $\mathcal{C}(D,G)^\bot =(c,\cdots, c) \cdot \mathcal{C}(D,H)=\mathcal{C}(D,H)$.

Suppose that $\mathrm{g.c.d}(G,H)$ is non-special.
$\mathcal{ C}(D,G)$ is a LCD code from  similar arguments used in proving Theorem \ref{gconstruction for curves}.

Suppose that $\mathcal{ C}(D,G)$ is a LCD code. If $\mathrm{g.c.d}(G,H)$ is special. Then, $l(\mathrm{g.c.d}(G,H))>0$. Let $0\neq f\in \mathcal L(\mathrm{g.c.d}(G,H))=\mathcal L(G) \cap \mathcal L(H)$. Thus,
$(f(P_1), \cdots, f(P_n))\in \mathcal{ C}(D,G) \cap  \mathcal{ C}(D,H)$. Note that $(f(P_1), \cdots, f(P_n))\neq (0, \cdots, 0)$,
which
contradicts $\mathcal{ C}(D,G) \cap  \mathcal{ C}(D,H)=\{0\}$. This  completes the proof.
\end{proof}

\begin{corollary}
Let $C$ be the projective line over $\mathbb F_q$, $\mathcal O$ be the   point
at infinity, $P$ be the original point and $D=C(\mathbb F_q) \backslash \{\mathcal O, P\}$. Let $G=r \mathcal O+ r P$ with $0<r \le \frac{q-2}{2}$. Then, $\mathcal{C}(D, G)$ is a maximum distance separable (MDS) LCD code over $\mathbb F_q$ with parameters $[q-1, 2r+1, n-2r]$. Moreover, $\mathcal{C}(D, G)$ has generator matrix
\begin{equation*}
\left(
\begin{array}{ccccc}
   1 & 1 & 1 & \cdots & 1 \\
   1 & \rho^{1} & \rho^{2} & \cdots & \rho^{q-2}\\
   1 & \rho^{-1} & \rho^{-2} & \cdots & \rho^{-(q-2)}\\
   \vdots & \vdots & \vdots &\ddots & \vdots \\
   1 & \rho^{i\cdot 1} & \rho^{i\cdot 2} & \cdots & \rho^{i\cdot (q-2)}\\
  1 & \rho^{-i\cdot 1} & \rho^{-i\cdot 2} & \cdots & \rho^{-i\cdot (q-2)}\\
   \vdots & \vdots & \vdots &\ddots & \vdots \\
    1 & \rho^{r\cdot 1} & \rho^{r\cdot 2} & \cdots & \rho^{r\cdot (q-2)}\\
  1 & \rho^{-r\cdot 1} & \rho^{-r\cdot 2} & \cdots & \rho^{-r\cdot (q-2)}
   \end{array}
\right),
\end{equation*}
where $\rho \in \mathbb F_q$ is a primitive element.
\end{corollary}
\begin{proof}
Let $H=(q-r-2)\mathcal O -(r+1)P$. Then, $\mathrm{g.c.d}=r \mathcal O -(r+1) P$ is non-special and $(\omega)=G+H-D$ with $\omega=\frac{1}{x^q-x} dx$.
Note that $\mathrm{Res}_{Q}(\omega)=-1$ for any $Q\in Supp(D)$. Form Theorem \ref{Res:const},  $\mathcal{C}(D, G)$ is a LCD code. The parameters
of $\mathcal{C}(D, G)$ follows from Theorem \ref{[n,k,d]} and the Singleton bound. Observe that $\{x^i: -r\le i \le r\}$ is a base of $\mathcal L(r \mathcal O+ r P)$. This completes the proof.
\end{proof}

\subsection{LCD codes from hyperelliptic curves}
Let $q=2^m$ and $C$ be the curve over $\mathbb F_{q^2}$ defined as
\begin{align*}
y^2+y=x^{q+1}.
\end{align*}
This curve has genus $g=\frac{q}{2}$. For any $\alpha\in \mathbb F_{q^2}$, there exactly exist two rational points $P_{\alpha}^+, P_{\alpha}^-$ with $x$-component $\alpha$.
Let $\mathcal O$ be the point at infinity. Then, the set $C(\mathbb F_{q^2})$ of all rational points of $C$ equal $\{P_{\alpha}^+: \alpha \in \mathbb F_{q^2}\}\cup \{P_{\alpha}^-: \alpha \in \mathbb F_{q^2}\} \cup \{\mathcal O\}$. Thus, $C$ has exactly $1+2q^2=1+q^2+2g\sqrt{q^2}$ rational points, which attains the well-known Hasse-Weil bound. Let $\omega=\frac{1}{x^{q^2}+x} dx$, then
\begin{align}\label{C}
(\omega)=2(q^2-1+\frac{q}{2}) \cdot \mathcal O- \sum_{\alpha\in \mathbb F_{q^2}}(P_{\alpha}^+ + P_{\alpha}^-) \text{ and } \mathrm{Res}_{P_{\alpha}^+}(\omega)=\mathrm{Res}_{P_{\alpha}^-}(\omega)=1,
\end{align}
for any $\alpha\in \mathbb F_{q^2}$.

\begin{theorem}\label{(r+g)O+rQ}
Let $q\ge 4$, $P$ be an affine point of $C$, $D=C(\mathbb F_{q^2}) \backslash \{\mathcal O,P\}$ and $G=(r+\frac{q}{2})\mathcal O+r P$, such that $\frac{q}{4}\le r \le q^2-\frac{q}{4}-1$ and $(r+\frac{q}{2})\mathcal O -(r+1)P$ is a non-special divisor. Then, $\mathcal{ C}(D,G)$ is a LCD code with dimension $2r+1$ and minimum distance $d\ge 2q^2-\frac{q}{2}-2r-1$,  and  the dual code $\mathcal{ C}(D,H)$ of $\mathcal{ C}(D,G)$  is a LCD code with dimension $2(q^2-r-1)$ and minimum distance $d^\bot \ge 2r-\frac{q}{2}+2$, where $H=(2q^2+\frac{q}{2}-r-2)\mathcal{O}-(r+1) P$.
\end{theorem}
\begin{proof}
From Equation (\ref{C}), $(\omega)=G+H-D$, where $\omega=\frac{1}{x^{q^2}+x}dx$. Observe that $\mathrm{g.c.d}(G,H)=(r+\frac{q}{2})\mathcal O-(r+1)P$. This theorem follows from  $(r+\frac{q}{2})\mathcal O -(r+1)P$ being a non-special divisor and Theorem \ref{gconstruction for curves}.
\end{proof}

\begin{example}
Let $q=2^2$ and $C$ be the genus $2$ hyperelliptic curve defined by $y^2+y=x^{q+1}$ over $\mathbb F_{q^2}$. Let $P=(0,0)$
 and $D=C(\mathbb F_{q^2}) \backslash \{\mathcal O, P\}$. Then $9\mathcal O-8P$ is non-special and $\# D=31$.
Let  $G=9\mathcal O+7P$ and $H=25\mathcal O-8P$. Then, $\mathcal{ C}(D,G)$ in Theorem \ref{(r+g)O+rQ} is a LCD code with parameters [31, 15], and  the dual code $\mathcal{ C}(D,H)$ in Theorem \ref{(r+g)O+rQ}  of $\mathcal{ C}(D,G)$  is a LCD code with parameters [31, 16], which is verified by MAGMA.
\end{example}

\begin{theorem}\label{g-1 non-special}
Let $q\ge 4$, $\sum_{i=1}^t n_i =g$ with $n_i> 0$ and $\sum_{i=1}^{t} r_i\le \frac{1}{4}(2q^2-\frac{3}{2}q-4t-4)$ with $r_i\ge 0$. Let $\{\alpha_1,\cdots, \alpha_t\}$ be a subset of $\mathbb F_{q^2}$ with cardinality $t$,  $D=C(\mathbb F_{q^2}) \backslash \{\mathcal O, P_{\alpha_1}^+, P_{\alpha_1}^-,
 \cdots,  $
$ P_{\alpha_t}^+, P_{\alpha_t}^-  \}$ and $G=(2(t+\sum_{i=1}^t r_i ) +q-1)\cdot \mathcal O+\sum_{i=1}^t r_i \cdot (P_{\alpha_i}^+ +P_{\alpha_i}^-)+\sum_{i=1}^t n_i \cdot P_{\alpha_i}^+ $. Then, $\mathcal{ C}(D,G)$ is a LCD code with dimension $4\sum_{i=1}^t r_i +2t+q$ and minimum distance $d\ge 2q^2-4\sum_{i=1}^t r_i -4t- \frac{3}{2}q   +1$,  and  the dual code $\mathcal{ C}(D,H)$ of $\mathcal{ C}(D,G)$  is a LCD code with dimension $2q^2-4\sum_{i=1}^t r_i -4t- q$ and minimum distance $d^\bot\ge  4\sum_{i=1}^t r_i +2t+\frac{1}{2} q +1$, where $H=(2q^2-2(t+\sum_{i=1}^t r_i )-1)\cdot \mathcal O-\sum_{i=1}^t (r_i +n_i+1) \cdot (P_{\alpha_i}^+ +P_{\alpha_i}^-)+\sum_{i=1}^t n_i \cdot P_{\alpha_i}^- $.
\end{theorem}
\begin{proof}
From Equation (\ref{C}), $(\omega)=G+H-D$, where $\omega=\frac{1}{x^{q^2}+x}dx$. Note that $\mathrm{g.c.d}(G,H)=(2(t+ \sum_{i=1}^{t} r_i) +q-1)\cdot \mathcal O   -\sum_{i=1}^{t}  (r_i+n_i+1)(P_{\alpha_i}^+  +P_{\alpha_i}^-)  +\sum_{i=1}^{t}  n_i  P_{\alpha_i}^-$.
From $ P_{\alpha_{i}}^+  +  P_{\alpha_{i}}^-  \sim 2\mathcal O$, we have $\mathrm{g.c.d}(G,H)= \sum_{i=1}^{t} n_i P_{\alpha_i}^-   -\mathcal O$. Since $\sum_{i=1}^{t} n_i P_{\alpha_i}^-$ is a reduced divisor, $l(\sum_{i=1}^{t} n_i P_{\alpha_i}^-)=1$. Thus $l(\mathrm{g.c.d}(G,H))=l(\sum_{i=1}^{t} n_i P_{\alpha_i}^-   -\mathcal O)=0$.
This theorem follows from   Theorem \ref{gconstruction for curves}.
\end{proof}

\begin{remark}
From $(x+\alpha_i)=P_{\alpha_i}^+  +P_{\alpha_i}^- - 2\mathcal O$,
$G=(\prod_{i=1}^t (x+\alpha_i)^{r_i})+ (4\sum_{i=1}^t r_i   +2t   +q-1)\cdot \mathcal O  +\sum_{i=1}^t  n_i  P_{\alpha_i}^+$.
Thus, $\mathcal L(G) =\frac{1}{\prod_{i=1}^t (x+\alpha_i)^{r_i}}   \mathcal L((4\sum_{i=1}^t r_i   +2t   +q-1)\cdot \mathcal O  +\sum_{i=1}^t  n_i  P_{\alpha_i}^+)$.
From the  similar discussion as above, one gets $\mathcal L(H)=\prod_{i=1}^t (x+\alpha_i)^{r_i+n_i+1}   \mathcal L((2q^2-q-\sum_{i=1}^t r_i-4t-1)\cdot \mathcal O    +\sum_{i=1}^t  n_i  P_{\alpha_i}^-)$. Let $D=\{P_{\beta_1}^+,P_{\beta_1}^-,\cdots  ,P_{\beta_{q^2-t}}^+,P_{\beta_{q^2-t}}^-\}$. Then,  $\mathcal{ C}(D,G)= \mathcal{ GC}(D,(4\sum_{i=1}^t r_i   +2t   +q-1)\cdot \mathcal O  +\sum_{i=1}^t  n_i  P_{\alpha_i}^+, \mathbf a)$ and $\mathcal{ C}(D,H)= \mathcal{ GC}(D,(2q^2-q-\sum_{i=1}^t r_i-4t-1)\cdot \mathcal O    +\sum_{i=1}^t  n_i  P_{\alpha_i}^-, \mathbf b)$, where $\mathbf a=(a_1,a_1,\cdots, a_{q^2-t},a_{q^2-t})$ and $\mathbf b=(b_1,b_1,\cdots, b_{q^2-t},b_{q^2-t})$,  $a_j=\frac{1}{\prod_{i=1}^t (\beta_j+\alpha_i)^{r_i}   }$,  and $b_j=\prod_{i=1}^t (\beta_j+\alpha_i)^{r_i+n_i+1} $.
\end{remark}
\begin{example}
Let $q=2^3$ and $C$ be the genus $4$ hyperelliptic curve defined by $y^2+y=x^{q+1}$ over $\mathbb F_{q^2}$. Let $P^+=(0,0)$
, $P^-=(0,1)$ and $D=C(\mathbb F_{q^2}) \backslash \{\mathcal O, P^+,P^-\}$. Then $\# D=126$.
Let  $G=19\mathcal O+5(P^+ +P^-)+ 4P^+$ and $H=115\mathcal O-10(P^+ + P^-) +P^-$. Then, $\mathcal{ C}(D,G)$ in Theorem \ref{g-1 non-special} is a LCD code with parameters [126, 30], and  the dual code $\mathcal{ C}(D,H)$ in Theorem \ref{g-1 non-special}  of $\mathcal{ C}(D,G)$ is a LCD code with parameters [126, 96], which is verified by MAGMA.
\end{example}

\begin{remark}
All constructions presented in Theorem \ref{(r+g)O+rQ} and Theorem \ref{g-1 non-special} can be directly generalized to any hyperelliptic curves.
\end{remark}

\subsection{LCD codes from Hermitian curves}
 Let $q$ be a power of any prime and  $C_{as}$ be the Hermitian curve over $\mathbb F_{q^2}$ defined by
 $$y^q+y=x^{q+1}.$$
Then $C_{as}$ is also an Artin-Schreier curve. The curve $C_{as}$ has genus $g=\frac{1}{2}q(q-1)$, and for every $\alpha\in \mathbb F_{q^2}$ the element $x-\alpha$ has $q$ zeros of degree one in $C_{as}$. Except the   point $\mathcal O$ at infinity, all rational points of $C_{as}$ are obtained in this way. One easily checks that the Hasse-Weil bound is attained. Let $\omega=\frac{1}{x^{q^2}-x} dx$, then
 \[
 (\omega)=(n+2g-2)\mathcal O -D,
 \]
 where $D=C_{as}(\mathbb F_{q^2})\backslash \{\mathcal O\}$ and $n=\# D=q^3$. Then, $\mathrm{Res}_{P}(\omega)=-1$.
We refer to \cite{Sti} for more details about Hermitian curves.
\begin{theorem}\label{AS}
Let $C_{ad}, g, \mathcal O, D$ and $\omega$ be defined as before.
Let $G=(r \cdot deg(P)+g-1)\cdot \mathcal O+ r\cdot P$,  where $P$ is a place of $C_{as}$ with degree more than $1$ and $r$ is a positive integer whit $r\cdot deg(P)\le \frac{n}{2}$.
Then, $\mathcal C(D,G)$ is a LCD codes if and only if $(r \cdot deg(P)+g-1)\cdot \mathcal O- r\cdot P$ is non-special.
\end{theorem}
\begin{proof}
Let $H=(n-r\cdot deg(P)+g-1)\cdot \mathcal O -r\cdot P$. Then $\mathrm{g.c.d}(G, H)=(r \cdot deg(P)+g-1)\cdot \mathcal O- r\cdot P$, $(\omega)=G+H-D$ and $\mathrm{Res}_{Q}(\omega)=-1$ for any $Q\in Supp(D)$. From Theorem \ref{Res:const}, this theorem follows.
\end{proof}

\begin{example}
Let $q=3$ and $C$ be the genus $3$ Hermitian curve defined by $y^q+y=x^{q+1}$ over $\mathbb F_{q^2}$. Let $P$ be the degree $3$ place at $(\beta, \rho^2 \beta^2+\beta-1)$
, where $\rho \in \mathbb{F}_9, \beta\in \mathbb{F}_{9^3}$, $\rho^2-\rho-1=0$, and $\beta^3+\rho \beta^2-\beta+\rho^2=0$. Let $D=C(\mathbb F_{q^2}) \backslash \{\mathcal O\}$ and $G=8\mathcal O+2P$.
Then, $\# D=27$ and $8\mathcal O-2P$ is non-special. $\mathcal{ C}(D,G)$ in Theorem \ref{AS} is a LCD code with parameters [27, 12], which is verified by  MAGMA.
\end{example}

\section{conclusion}
This paper is devoted to the construction  of particular AG complementary dual (LCD) codes which can be resistant against side-channel attacks (SCA).
We firstly provide a construction scheme  for obtaining LCD codes from elliptic curves and present some explicit LCD codes from elliptic curves, which contain some infinite class of optimal codes with parameters meeting Griesmer bound on linear codes.
All codes constructed from elliptic curve are MDS or almost MDS. We also introduce a  construction mechanism for obtaining LCD codes from any algebraic
curve and derive some explicit LCD codes from hyperelliptic curves and Hermitian curves. In a future work, we will study the resistance of algebraic geometry LCD codes to SCA.


\ifCLASSOPTIONcaptionsoff
  \newpage
\fi

\end{document}